\documentclass[a4paper,11pt]{article}

\usepackage[english]{babel}
\usepackage[utf8]{inputenc}
\usepackage[T1]{fontenc}
\usepackage[usenames, dvipsnames]{xcolor}
\usepackage{tikz}
\usepackage{amsthm}
\usepackage{amsmath,amssymb}
\usepackage{fullpage}
\usepackage{hyperref}
\usepackage{xspace,float,caption}
\usepackage[inner=2.5cm,outer=2.5cm, top=2.5cm, bottom=2.5cm]{geometry}
\usepackage[color=green!30]{todonotes}

\newtheorem{theorem}{Theorem}

\newtheorem{lemma}[theorem]{Lemma}
\newtheorem{corollary}[theorem]{Corollary}

\newcommand{\nats}{\mathbb{N}}
\newcommand{\pnats}{\nats_{>0}}
\newcommand{\reals}{\mathbb{R}}
\newcommand{\bs}[1]{\boldsymbol{#1}}
\newcommand{\C}{\mathcal{C}}
\newcommand{\ds}{\displaystyle}
\newcommand{\eqdef}{\stackrel{\textrm{def}}{=}}
\newcommand{\yhat}{\hat y}
\newcommand{\Hol}{\mathrm{Holant}}
\newcommand{\PlHol}{\mathrm{Pl\mbox{-}Holant}}
\newcommand{\numP}{\mathrm{\#P}}
\newcommand{\sigmahat}{\sigma}
\newcommand{\sigmatilde}{\tilde\sigma}
\newcommand{\tautilde}{\tilde\tau}
\newcommand{\what}{\hat w}
\newcommand{\calE}{\mathcal{E}}
\newcommand{\hol}{\mathcal{H}}
\newcommand\calD{\mathcal{D}}
\newcommand\pitilde{\tilde{\pi}}
\newcommand{\rhotilde}{\tilde{\rho}}
\newcommand{\Ptilde}{\widetilde{P}}
\newcommand{\Qtilde}{\widetilde{Q}}

\let\rho=\varrho
\let\epsilon=\varepsilon

\title{Polynomial-time approximation algorithms\\ for the antiferromagnetic Ising model on line graphs}
\author{
Martin Dyer\thanks{Work supported by
    EPSRC grants EP/S016562/1 and EP/S016694/1, ``Sampling in hereditary classes''.}\\
\small School of Computing\\[-0.5ex]
\small University of Leeds\\[-0.5ex]
\small Leeds LS2~9JT, UK\\[-0.5ex]
\small\texttt{m.e.dyer@leeds.ac.uk}\\
\and Marc Heinrich\footnotemark[1]\\
\small School of Computing\\[-0.5ex]
\small University of Leeds\\[-0.5ex]
\small Leeds LS2~9JT, UK\\[-0.5ex]
\small\texttt{m.heinrich@leeds.ac.uk}
\and Mark Jerrum\footnotemark[1]\\
\small School of Mathematical Sciences\\[-0.5ex]
\small Queen Mary University of London\\[-0.5ex]
\small Mile End Road, London E1 4NS, UK\\[-0.5ex]
\small\texttt{m.jerrum@qmul.ac.uk}
\and Haiko M\"{u}ller\footnotemark[1]\\
\small School of Computing\\[-0.5ex]
\small University of Leeds\\[-0.5ex]
\small Leeds LS2~9JT, UK\\[-0.5ex]
\small\texttt{h.muller@leeds.ac.uk}
}

\begin{document}

\maketitle

\begin{abstract}
We present a polynomial-time Markov chain Monte Carlo algorithm for estimating the partition function of the  antiferromagnetic Ising model on any line graph.  The analysis of the algorithm exploits the ``winding'' technology devised by McQuillan [CoRR abs/1301.2880 (2013)] and developed by Huang, Lu and Zhang [Proc.\ 27th Symp.\ on Disc.\ Algorithms (SODA16), 514--527].  We show that exact computation of the partition function is \#P-hard, even for line graphs, indicating that an approximation algorithm is the best that can be expected.  We also show that Glauber dynamics for the Ising model is rapidly mixing on line graphs, an example being the kagome lattice. 
\end{abstract}

\section{Introduction}

In statistical mechanics, the Ising model was proposed by Lenz in 1920 as a mathematical model of ferromagnetism.  The model is defined in terms of simple local interactions between sites representing atoms.  Despite the simplicity of its description, the model exhibits complex behaviours, which have been hard to unravel.  Indeed several years were to pass before the most fundamental property of the Ising model was established, namely that small changes in the strength of the local interactions could lead to dramatic changes in macroscopic behaviour.  The model has remained an active object of study ever since, with applications to statistical inference and image analysis helping to sustain interest~\cite{WJ08}.  An accessible account of the Ising model can be found in Friedli and Velenik's book~\cite[\S3]{FriedliVelenik}.

For $\beta,\nu\in\reals$, the partition function of the Ising model on a (simple, undirected) graph~$\Gamma$ is defined by
\begin{align}
Z_{\beta,\nu}(\Gamma)&=\sum_{\sigma:V(\Gamma)\to\{0,1\}}w(\sigma)\label{eq:Isingpf}\\
\intertext{where}
w(\sigma)&=\prod_{\{i,j\}\in E(\Gamma)}\exp\big(\beta\,|\sigma(i)-\sigma(j)|\big)\prod_{k\in V(\Gamma)}\exp(\nu\sigma(k)).\notag
\end{align}
The partition function is a sum over \emph{spin configurations} $\sigma:V(\Gamma)\to\{0,1\}$, which assign 0/1 ``spins'' to the vertices of~$\Gamma$.  The parameter $\beta$ is an ``interaction energy'' that controls the strength of the interaction between spins at adjacent vertices.  The parameter $\nu$ indicates the strength of the ``external field'' that biases the spins at individual vertices.  The partition function is the normalising factor in a probability distribution $\mathcal{D}_{\Gamma,\beta,\nu}$, the \emph{Gibbs distribution}, that assigns weight $\mathcal{D}_{\Gamma,\beta,\nu}(\sigma)=w(\sigma)/Z_{\beta,\nu}(\Gamma)$ to spin configuration~$\sigma$.
More generally, the uniform interaction energy~$\beta$ may be replaced by individual interaction energies $\beta_{i,j}$ for each edge $\{i,j\}\in E(\Gamma)$, and the uniform field strength $\nu$ by individual field strengths~$\nu_k$ for each vertex $k\in V(\Gamma)$.  For our main result, it is crucial that the interaction energies are uniform over edges, but uniformity of the field is inessential.  For simplicity, we carry through the calculation for the uniform case, and return to varying fields towards the end, in Section~\ref{sec:pos}.  

In this article, we analyse the computational complexity of evaluating the partition function~\eqref{eq:Isingpf} when $\Gamma$ is a line graph and $\beta>0$.  We also address the related problem of sampling configurations according to the Gibbs distribution $\mathcal{D}_{\Gamma,\beta,\nu}$.  The class of line graphs, which is well-studied in graph theory, will be defined presently.  When $\beta>0$ the spins $\sigma(i)$ and $\sigma(j)$ associated with adjacent vertices $i$ and~$j$ tend to differ;  that is, configurations $\sigma:V(\Gamma)\to\{0,1\}$ with $\sigma(i)\not=\sigma(j)$ make a greater contribution to the sum~\eqref{eq:Isingpf}.  The qualifier \emph{antiferromagnetic} is applied to this situation, while \emph{ferromagnetic} describes the case $\beta<0$.\footnote{This notation is slightly nonstandard, in that our interaction energy $\beta$ has the opposite sign to usual, but this convention becomes convenient later on.}  When $\beta=0$ the spins are probabilistically independent and the model is trivial.

We now briefly describe the context for our investigation.  Exact evaluation of $Z_{\beta,\nu}(\Gamma)$ is generally \#P-hard, and hence almost certainly computationally intractable \cite[Thm 15]{JerSinIsing}.  The main exception is when the graph $\Gamma$ is restricted to be planar and the external field is absent ($\nu=0$), a situation which is covered by a classical algorithm due to Fisher~\cite{Fisher} and Kasteleyn~\cite{Kasteleyn}.  (Planarity of $\Gamma$ can be relaxed somewhat, to the condition that $\Gamma$ admits a Pfaffian orientation.)  

Turning to approximate computation, Jerrum and Sinclair~\cite{JerSinIsing} present an efficient algorithm to estimate (in the sense of Fully Polynomial Randomised Approximation Scheme or FPRAS) the partition function $Z_{\beta,\nu}(\Gamma)$ for general graphs~$\Gamma$ in the ferromagnetic case, even when $\nu\not=0$.  This algorithm extends to the situation of varying fields, provided that the sign of $\nu_k$ is consistent, but breaks down if the sign varies.  In general, the partition function is hard to approximate in the antiferromagnetic case.  The reason, intuitively, is that the ground states of the model, i.e., the configurations of greatest weight, correspond to maximum cardinality cuts in~$\Gamma$.  Thus, an efficient approximation algorithm for the partition function would most likely yield a polynomial-time algorithm for finding a maximum cut in a graph, which is an NP-hard problem~\cite{GJS76}.  So, to find tractable classes of instances in the antiferromagnetic situation, we need to focus on graph classes in which a maximum cardinality cut can be found in polynomial time.  One obvious example is the class of bipartite graphs.  Graphs in this class are easily dealt with:  just invert the sign of~$\beta$ and flip the interpretation of the spins on one side of the bipartition to yield an equivalent ferromagnetic Ising model.  For this we have an FPRAS, as mentioned earlier.  

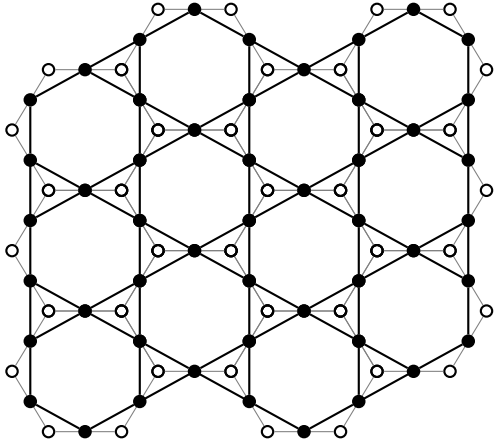
\begin{figure}[t]
\centering
\begin{tikzpicture}[xscale=0.08, yscale=0.08, inner sep=1.5pt, >=stealth]
\tikzset{hex/.style={circle,thick,draw}}
\tikzset{kag/.style={circle,thick,draw,fill}}

\foreach \x in {0,36}
\foreach \y in {0,20,40} {
   \draw (\x,\y) + (0,0)    node[hex] (v1) {};
   \draw (\x,\y) + (6,10)   node[hex] (v2) {};
   \draw (\x,\y) + (18,10)  node[hex] (v3) {};
   \draw (\x,\y) + (24,0)   node[hex] (v4) {};
   \draw (\x,\y) + (18,-10) node[hex] (v5) {};
   \draw (\x,\y) + (6,-10)  node[hex] (v6) {};

   \draw[gray] (v1) -- (v2) -- (v3) -- (v4) -- (v5) -- (v6) -- (v1);
   
   \draw (\x,\y) + (3,5)    node[kag] (v12) {};
   \draw (\x,\y) + (12,10)  node[kag] (v23) {};
   \draw (\x,\y) + (21,5)   node[kag] (v34) {};
   \draw (\x,\y) + (21,-5)  node[kag] (v45) {};
   \draw (\x,\y) + (12,-10) node[kag] (v56) {};
   \draw (\x,\y) + (3,-5)   node[kag] (v61) {};

   \draw[thick] (v12) -- (v23) -- (v34) -- (v45) -- (v56) -- (v61) -- (v12);
}

\foreach \x in {18,54}
\foreach \y in {10,30,50} {
   \draw (\x,\y) + (0,0)    node[hex] (v1) {};
   \draw (\x,\y) + (6,10)   node[hex] (v2) {};
   \draw (\x,\y) + (18,10)  node[hex] (v3) {};
   \draw (\x,\y) + (24,0)   node[hex] (v4) {};
   \draw (\x,\y) + (18,-10) node[hex] (v5) {};
   \draw (\x,\y) + (6,-10)  node[hex] (v6) {};

   \draw[gray] (v1) -- (v2) -- (v3) -- (v4) -- (v5) -- (v6) -- (v1);
   
   \draw (\x,\y) + (3,5)    node[kag] (v12) {};
   \draw (\x,\y) + (12,10)  node[kag] (v23) {};
   \draw (\x,\y) + (21,5)   node[kag] (v34) {};
   \draw (\x,\y) + (21,-5)  node[kag] (v45) {};
   \draw (\x,\y) + (12,-10) node[kag] (v56) {};
   \draw (\x,\y) + (3,-5)   node[kag] (v61) {};

   \draw[thick] (v12) -- (v23) -- (v34) -- (v45) -- (v56) -- (v61) -- (v12);
}

\draw[thick] (12,50) -- (21,55);
\draw[thick] (39,55)-- (48,50) -- (57,55);

\draw[thick] (21,-5) -- (30,0) -- (39,-5);
\draw[thick] (57,-5) -- (66,0);

\draw[thick] (3,5) -- (3,15);
\draw[thick] (3,25) -- (3,35);

\draw[thick] (75,15) -- (75,25);
\draw[thick] (75,35) -- (75,45);

\end{tikzpicture}

\caption{A portion of the hexagonal lattice (open vertices/thin grey edges) and its line graph, the kagome lattice (filled vertices/thick black edges)}
\label{fig:kagome}
\end{figure}

Here, we focus on another graph class for which maximum cardinality cut is polynomial-time solvable, namely the class of line graphs.  The \emph{line graph} $L(\Gamma)$ associated with an underlying graph $\Gamma=(V(\Gamma),E(\Gamma))$ is the graph with vertex set $V(L(\Gamma))=E(\Gamma)$ and edge set 
$$E(L(\Gamma))=\big\{\{e,f\}:e,f\in E(\Gamma)\text{ and }e\cap f\not=\emptyset\big\}.$$
Thus, there is an edge between two vertices in the line graph $L(\Gamma)$ whenever the corresponding edges in~$\Gamma$ share an endpoint.  See Figure~\ref{fig:kagome} for an example graph and its associated line graph.  That a maximum cardinality cut can found in polynomial time in line graphs was shown by Arbib~\cite{Arbib88};  see also Guruswami~\cite{Guruswami}.  Our main contribution (Theorem~\ref{thm:FPRAS}) is an efficient approximation algorithm for estimating the partition function $Z_{\beta,\nu}(\Gamma)$ in the antiferromagnetic case, when $\Gamma$ is a line graph.  The algorithm is efficient in the sense of Fully Polynomial Randomised Approximation Scheme (FPRAS).  An FPRAS is a randomised algorithm that produces an output that is within ratio $1\pm\epsilon$ of the desired solution, with high probability, and within time that is bounded by a polynomial in the instance size~$n$ and $1/\epsilon$.  See \cite[Defn~11.2]{MotRag} for details.  Along the way, we also find an efficient algorithm (Theorem~\ref{thm:FPAUS}) for sampling configurations in the same setting.  In contrast to the case of bipartite graphs, the reason for tractability of these two problems, counting and sampling, is non-trivial.  Our precise results are the following.   

\begin{theorem}\label{thm:FPRAS}
Suppose $\beta>0$ (the antiferromagnetic case) and $\nu\in\reals$. 
There is an FPRAS that, given a line graph $L(\Gamma)$, estimates the partition function $Z_{\beta,\nu}(L(\Gamma))$ of the Ising model~\eqref{eq:Isingpf}.
\end{theorem} 

\begin{theorem}\label{thm:FPAUS}
Suppose $\beta>0$ (the antiferromagnetic case) and $\nu\in\reals$. 
There is an algorithm that, given a line graph $L(\Gamma)$ and tolerance $\delta>0$, produces a spin configuration $\sigma$ from a distribution that is within variation distance $\delta$ of the Gibbs distribution $\mathcal{D}_{L(\Gamma),\beta,\nu}$.  The running time of the algorithm is polynomial in $\log\delta^{-1}$ and the instance size.
\end{theorem} 

One consequence of our results (see Section~\ref{sec:pos} for details) is that Glauber dynamics for the antiferromagnetic Ising model on a line graph is rapidly mixing.  Glauber dynamics moves around the space of configurations by making single spin updates according to a simple probabilistic rule.  Its unique stationary distribution is the Gibbs distribution defined earlier.  The phrase ``rapidly mixing'' indicates that the dynamics converges close to the Gibbs distribution in a number of updates that is polynomial in the size of the graph.  In a recent paper, Chen, Liu and Vigoda~\cite{chen2020rapid} show that Glauber dynamics mixes rapidly up to a natural threshold for~$\beta$, depending on the maximum degree of the graph, beyond which it was already known that the mixing time is exponential.  Here we show that, by suitably restricting the underlying graph, we can have rapid mixing for all~$\beta$. 

Rapid mixing at all $\beta$ is perhaps a surprising phenomenon in the context of the Ising model, and has an interesting relationship to phase transitions in infinite systems.  As a concrete example, take the kagome lattice, which is the line-graph of the familiar hexagonal lattice.  Figure~\ref{fig:kagome} depicts a portion of the kagome lattice superimposed on the underlying hexagonal lattice.  Consider an $L\times L$ patch of the kagome lattice with top and bottom, and left and right side identified (toroidal or periodic boundary conditions).  This finite graph is a line graph, so Glauber dynamics mixes in polynomial time on this graph.  Heuristically, this ought to imply that the antiferromagnetic Ising model on the infinite kagome lattice does not exhibit a phase transition.  In other words, as the interaction strength~$\beta$ is continuously varied, no dramatic change occurs in the typical macroscopic configurations.  Indeed Sy\^{o}zi~\cite{Syozi} has shown that the antiferromagnetic Ising model on the kagome lattice does not exhibit a phase transition at any $\beta>0$ (which is taken to mean that a certain observable called the specific heat is continuous in~$\beta$).  This is another example of the emerging connection between algorithmic tractability of finite systems and absence of a phase transition in the corresponding infinite system.  We leave the connection to phase transitions as a heuristic observation as there are currently no rigorous results in this direction.  Indeed, Goldberg, Martin and Paterson~\cite{GMP} show polynomial-time mixing of Glauber dynamics on 3-colourings of finite portions of the square lattice, and yet the infinite lattice has ``frozen'' colourings from which no transition is possible.  

We note that the kagome lattice is not the only line graph that has been studied in the context of statistical physics.  Another example, this time 3-dimensional, is the pyrochlore lattice, which is the line graph of a degree-4 graph:  see Jur\v{c}i\v{s}inov\'{a} and Jur\v{c}i\v{s}in~\cite{JJ20}.
     
As the kagome lattice exhibits a phase transition in the ferromagnetic situation, we cannot expect Glauber dynamics to be rapidly mixing for all $\beta<0$, even on line graphs. However, our approach does show rapid mixing for $\beta_c\leq\beta\leq0$, for some $\beta_c<0$ depending on the maximum degree~$\Delta$ of~$\Gamma$. This is so at least for small $\Delta$ by direct computation, and probably for all $\Delta$, though we have not shown this. See Section~\ref{sec:pos}.  

It is natural to question whether \emph{approximation} algorithms are necessary in the context of line graphs. Lemma~\ref{lem:zerofield} and Lemma~\ref{lem:planar} in Section~\ref{sec:neg} can be interpreted as asserting that we must settle for approximate results, assuming $\mathrm{RP}\not=\mathrm{NP}$: computing the partition function \emph{exactly} is hard even in the case of zero external field (Lemma~\ref{lem:zerofield}), or in the case of a planar graph (Lemma~\ref{lem:planar}).  As we have noted, imposing both conditions simultaneously leads to tractability.  The Ising model on the kagome lattice is exactly solvable~\cite{Syozi} with zero external field, but in the presence of an external field it seems we have to settle for the approximate solution promised by Theorem~\ref{thm:FPRAS}.

Finally, we can speculate on the existence of other graph classes on which the antiferromagnetic Ising model can be solved approximately.  The maximum cut problem is known to be NP-complete on many graph classes.  For example, Bodlaender and Jansen~\cite{BJ00} show that maximum cut is NP-hard on several graph classes including chordal graphs and tripartite graphs.  For such classes, it is unlikely that the antiferromagnetic Ising model is tractable, even in the approximate sense.  Indeed, whenever maximum cut cannot be solved to within a constant factor in polynomial time (under some complexity theoretic assumption) this heuristic claim becomes rigorous.   The general argument is given, in the context of maximum independent set, by Luby and Vigoda~\cite{LV99}.

Very recently, Bencs, Csikv\'ari and Regts~\cite{BCR20}, using a completely different approach to the one described here, gave a \emph{deterministic} algorithm for approximating the partition function of an antiferromagnetic model on a line graph.  

\section{Overview of the approach}\label{sec:overview}

We use an off-the-shelf analytical technique, known as ``winding'', devised by McQuillan~\cite{McQuil13} and developed and systematised by Huang, Lu and Zhang~\cite{HLZ16}.  As far as the proofs in this article are concerned, we can treat this technique for the most part as a black box.  Nevertheless, we do have to peer a little way inside the black box.  We'll  cover enough of the method for our current needs and leave the interested reader to consult the original sources for details.

The first step is to change our viewpoint away from vertex configurations towards edge configurations.  Recall that the spin configurations of the Ising model we are working with live on the \emph{vertices} of a line graph $L(\Gamma)$.  These configurations can equally be thought of as spin configurations living on the \emph{edges} of~$\Gamma$.  Such a spin model is called a ``holant'' in the computer science literature~\cite{CaiChenBook}.  In the holant view, the partition function of the Ising model has the form
\begin{equation}\label{eq:holant}
Z_{\beta,\nu}(L(\Gamma))=\sum_{\sigma:E(\Gamma)\to\{0,1\}}\,\prod_{k\in V(\Gamma)}f_k(\sigma|_{E(k)}),
\end{equation}
where $\sigma|_{E(k)}$ denotes the restriction of $\sigma$ to the edges incident at vertex~$k$.  The functions $\{f_k\}$ are specific to the Ising model and depend on the vertex degrees;  they will be derived in the next section.

The basic approach is Markov chain Monte Carlo (MCMC).  As usual, we concentrate on solving the sampling problem (Theorem~\ref{thm:FPAUS}) first, as it is a short step from there to a solution of the estimation problem (Theorem~\ref{thm:FPRAS}).  Our first thought would be to construct a Markov chain whose state space is the set of configurations $\sigma:E(\Gamma)\to\{0,1\}$.  However, the winding approach, as developed in~\cite{HLZ16}, requires us to expand the state space somewhat.  Think of each edge $\{i,j\}$ of $\Gamma$ as being subdivided into two \emph{half-edges} one incident at~$i$ and the other at~$j$.  Denote by $\calE=\calE(\Gamma)$ the set of all half-edges of~$\Gamma$.  Now consider an enlarged set of configurations $\sigmahat:\calE\to\{0,1\}$.  We say that a configuration $\sigmahat$ is \emph{consistent} if, for every edge, the two half-edges that compose it are assigned the same spin.  Let $\Omega_0$ denote the set of all consistent configurations, and define
$$
\hol_0(\Gamma)=\hol_0(\Gamma;\{f_k\})
=\sum_{\sigmahat\in\Omega_0}\,\prod_{k\in V(\Gamma)}f_k(\sigmahat|_{\calE(k)}),
$$
where $\sigmahat|_{\calE(k)}$ denotes the restriction of $\sigmahat$ to the half-edges incident at vertex~$k$.  Clearly $\hol_0(\Gamma)=Z_{\beta,\nu}(L(\Gamma))$.  
We say that a configuration $\sigmahat$ is \emph{nearly consistent} if it is consistent except for precisely two edges.  Denote by $\Omega_2$ be the set of all nearly consistent configurations, and define
$$
\hol_2(\Gamma)=\hol_2(\Gamma;\{f_k\})=
\sum_{\sigmahat\in\Omega_2}\,\prod_{k\in V(\Gamma)}f_k(\sigmahat|_{\calE(k)}).
$$
Nearly consistent configurations are needed in general to connect the state space of the Markov chain from~\cite{HLZ16} that we are about to define, but could be avoided in this specific application:  see Section~\ref{sec:Glauber}. 

We now consider a Markov chain on state space $\Omega=\Omega_0\cup\Omega_2$ in which a transition is available between every pair of configurations that differ on exactly two half-edges.   It is a simple matter to choose transition probabilities so that, in the stationary distribution, configuration~$\sigmahat$ occurs with probability proportional to $\what(\sigmahat)$, where 
$$
\what(\sigmahat)=\prod_{k\in V(\Gamma)}f_k(\sigmahat|_{\calE(k)}).
$$
Specifically, Huang et al.\ use transition probabilities $P:\Omega^2\to[0,1]$ based on the Metropolis filter.  Denote by $d(\sigmahat,\sigmahat')$ the number of half-edges on which $\sigmahat$ and $\sigmahat'$ differ, and let $m=|E(\Gamma)|$.  Then 
\begin{equation}\label{eq:Metropolis}
P(\sigmahat,\sigmahat')=\begin{cases}
\frac1{4m^2}\min\left\{\frac{\what(\sigmahat')}{\what(\sigmahat)},1\right\},&\text{if $d(\sigmahat,\sigmahat')=2$};\\
0,&\text{if $d(\sigmahat,\sigmahat')\notin\{0,2\}$},
\end{cases}
\end{equation}
and the loop probabilities $P(\sigmahat,\sigmahat)$ are chosen to satisfy $\sum_{\sigmahat'\in\Omega}P(\sigmahat,\sigmahat')=1$.
It may be verified that $P(\sigma,\sigma)>\frac12$;  this ``laziness'' property is convenient in the analysis of mixing time.\footnote{Note that the ``proposal probability'' $1/4m^2$ here is smaller than the one appearing in~\cite{HLZ16}, which is itself too large to guarantee $P(\sigmahat,\sigmahat)\geq0$.}  It is an immediate consequence of the Metropolis transition rule that the stationary distribution assigns probability proportional to $\what(\sigma)$ to each state $\sigma\in\Omega$. 

Note that, in particular, the stationary distribution, restricted to $\Omega_0$, is the Gibbs distribution $\mathcal{D}_{\Gamma,\beta,\nu}$ (under the obvious bijection between $\Omega_0$ and assignments $E(\Gamma)\to\{0,1\}$).  So we will have a polynomial-time approximate sampler for $\mathcal{D}_{\Gamma,\beta,\nu}$ if (a)~the Markov chain is rapidly mixing, i.e., comes close to stationarity in polynomially many steps, and (b)~the ratio $\hol_2(\Gamma)/\hol_0(\Gamma)$ is polynomially bounded.  Condition~(b) ensures that, when we stop the Markov chain after sufficiently many steps, we have a good chance of being at a state in~$\Omega_0$.  (Intriguingly, the ratio $\hol_2(\Gamma)/\hol_0(\Gamma)$ also plays a critical role in bounding the mixing time.)

Condition (b) can be swiftly dispensed with.

\begin{lemma}\label{lem:ratio}
$\hol_2(\Gamma)/\hol_0(\Gamma)=O(m^2)$, where $m=|E(\Gamma)|$.  The implicit constant in the O-notation depends on $\beta$, $\nu$ and the maximum degree $\Delta$ of\,~$\Gamma$.
\end{lemma}
\begin{proof}
Given a configuration $\sigmahat_2\in\Omega_2$, we can reach a configuration $\sigmahat_0\in\Omega_0$ by flipping the spin on two half-edges.  This operation changes the weight by a constant factor, i.e., $\what(\sigmahat_0)\geq C\what(\sigmahat_2)$ where $C$ depends on $\beta$, $\nu$ and $\Delta$.  Going backwards, there are $\binom{2m}2\leq2m^2$ half-edges on which we could choose to flip the spin.  
\end{proof}

To prove rapid mixing of the Markov chain we use the ``canonical paths'' method or, more accurately, its generalisation to ``multi-commodity flow''~\cite{flow}.  The key parameter of interest is the so-called congestion of the flow, as the mixing time scales linearly with congestion.  Generally speaking, the construction of appropriate canonical paths and the analysis of their congestion requires expertise and insight.  However, for the holant setting with 0/1 spins, this forbidding step can to a large extent be reduced to linear algebra by the winding technique.  In this approach, the selection of canonical paths follows automatically from the condition of ``windability'' of the functions~$f_k$.  The concept of windability of functions will be discussed a little later.  Then, the conductance of this set of canonical paths can be bounded in terms of the ratio appearing in Lemma~\ref{lem:ratio}.  Thus, in light of that lemma, the canonical paths method and associated concept of conductance can also be treated as a black box.

In principle, windability of any given function can be tested, at least for reasonable arities, with the aid of a computer algebra system.  The main obstacle we face in the current application is the need to demonstrate that the particular functions arising from the antiferromagnetic Ising model are windable for all arities (i.e., for all vertex degrees).  As will become apparent, overcoming this obstacle presents a challenge.  There are few natural infinite families of functions that are known to be uniformly windable, possibly just the examples given by McQuillan~\cite{McQuil13}, namely parity (Even and Odd) and not-all-equal (NAE).  So it is significant that the Ising model has this property.  

\section{The Ising model and holants}\label{sec:holants}

In this section we formalise the intuition that the partition function $Z_{\beta,\nu}(L(\Gamma))$ of the Ising model on a line graph $L(\Gamma)$ can be expressed as a holant on the underlying graph~$\Gamma$.

Denote the nonnegative real numbers by $\reals^+$ and the positive integers by $\pnats$. 
A symmetric function $F : \{0,1\}^d \rightarrow \reals^+$ of arity~$d$ has an associated \emph{signature vector} $[f_0, \ldots f_d]$ if $F(x) = f_i$ when $x$ has exactly $i$~coordinates set to~$1$. If $F$ is also \emph{self-complementary}, that is,  $F(x) = F(\bar x)$ for all $x$, then we have $f_i = f_{d-i}$. In that case, the vector $\bs{h}=[f_0, \ldots, f_{\lfloor d/2\rfloor}]$ completely characterises~$F$, and will be called the \emph{half-vector} of $F$. Henceforward, all functions $F$ will be symmetric.

For $\beta,\mu\in\reals$ and $d\in\pnats$, denote by $F_{\beta,\mu, d} : \{0,1\}^{d}\rightarrow \reals^+$ the function with signature vector $[e^{\beta i(d-i)+\mu i}:0\leq i \leq d]$. We will write $F_{\beta, d}$ for $F_{\beta,0, d}$. Note that $F_{\beta, d}$ is self-complementary, but $F_{\beta,\mu, d}$ is not, if $\mu\neq 0$.  
In the derivation below, we show that $Z_{\beta,\nu}(L(\Gamma))$ can be expressed in the form~\eqref{eq:holant}, in which, for each vertex $k$, we set $f_k=F_{\beta,\mu,d(k)}$, where $\mu=\nu/2$ and $d(k)$ is the degree of~$k$.
In what follows, $E(\Gamma)^{(2)}$ denotes the set of unordered pairs of edges of~$\Gamma$.
\begin{align*}
Z_{\beta,\nu}(L(\Gamma))&=\sum_{\sigma:V(L(\Gamma))\to\{0,1\}}\,\,\prod_{\{i,j\}\in E(L(\Gamma))}\exp\big(\beta\,|\sigma(i)-\sigma(j)|\big)\prod_{k\in V(L(\Gamma))}\exp(\nu\sigma(k))\\
&=\sum_{\sigma:E(\Gamma)\to\{0,1\}}\,\,\prod_{\substack{\{e,f\}\in E(\Gamma)^{(2)}\\e\cap f\not=\emptyset}}\exp\big(\beta\,|\sigma(e)-\sigma(f)|\big)\prod_{e\in E(\Gamma)}\exp(\nu\sigma(e))\displaybreak[1]\\
&=\sum_{\sigma:E(\Gamma)\to\{0,1\}}\,\,\prod_{k\in V(\Gamma)}\,\,\prod_{\substack{\{e,f\}\in E(\Gamma)^{(2)}\\e,f\ni k}}\exp\big(\beta\,|\sigma(e)-\sigma(f)|\big)\prod_{e\in E(\Gamma)}\exp(\nu\sigma(e))\displaybreak[1]\\
&=\sum_{\sigma:E(\Gamma)\to\{0,1\}}\,\,\prod_{k\in V(\Gamma)}F_{\beta,d(k)}\big(\sigma(e_{k,1}),\sigma(e_{k,2}),\ldots,\sigma(e_{k,d(k)})\big)\prod_{e\in E(\Gamma)}\exp(\nu\sigma(e))\\
&=\sum_{\sigma:E(\Gamma)\to\{0,1\}}\,\,\prod_{k\in V(\Gamma)}F_{\beta,\mu,d(k)}\big(\sigma(e_{k,1}),\sigma(e_{k,2}),\ldots,\sigma(e_{k,d(k)})\big),
\end{align*} 
where $e_{k,1},e_{k,2},\ldots,e_{k,d(k)}$ is an enumeration of edges incident at vertex $k$, and $\mu=\nu/2$. 

Comparing with \eqref{eq:holant}, it can be seen that we have successfully expressed the partition function $Z_{\beta,\nu}(L(\Gamma))$ as a holant on~$\Gamma$.  
 
\section{Windability}\label{sec:windability}
This section will be a proof of the following.
\begin{theorem}\label{thm:main}
  The function $F_{\beta,\mu, d}$ is windable for all $\beta\in\reals^+$, $\mu\in\reals$ and $d\in\nats$.
\end{theorem}
The definition of windable can be found in~\cite{McQuil13} or~\cite{HLZ16}.  However, a knowledge of the definition is not essential for an understanding of this paper, as we shall be using the characterisation of windable that was given by Huang et al., and that is captured in Theorem~\ref{thm01}.  In order to understand that theorem we need some preliminary definitions.

Suppose $m\in\pnats$ and let $n=\lfloor m/2\rfloor$. Recall from~\cite{HLZ16} the definition of the matrix $A_m \in \mathbb{R}^{(n+1) \times (n+1)}$, whose coefficients $a^{(m)}_{i,j}$ ($0\leq,i,j\leq n$) are, for $m=2n$ even,
\begin{align}\label{eq:Adefn}
	a^{(m)}_{i,j} &=
	\begin{cases}
		\ds {i \choose j} {2n - i \choose j} j! (i-j-1)!! (2n-i-j-1)!! & \text{if } i \equiv j \bmod 2\;,\\[3mm]
		\ds\ 0 & \text{otherwise}\;;
	\end{cases}\\
\intertext{or, for $m=2n+1$ odd,}
\label{eq:Adefn'}
a^{(m)}_{i,j} &=
\begin{cases}
\ds{i \choose j} {2n+1 - i \choose j} j! (i-j-1)!! (2n+1-i-j)!! & \text{if } i \equiv j \bmod 2\:, \\[3mm]
\ds {i \choose j} {2n+1 - i \choose j} j! (i-j)!! (2n-i-j)!! & \text{otherwise} \;,
\end{cases}
\end{align}
where $x!!$ is the double factorial $x!! := x \cdot (x-2) \cdot (x-4) \ldots$ for $x\in\nats$, and we use the convention $(-1)!!=1$.\footnote{The second formula has a typo in~\cite{HLZ16}, and is missing ``!!''.}

The matrix entries $a^{(m)}_{i,j}$ have a combinatorial meaning: the number of ways of pairing $m$ objects of two types, $i$ of the first type and $m-i$ of the second, such that exactly $j$ pairs have both types. There will be one singleton if $m=2n+1$. It follows that $\sum_{j=0}^{i}a^{(m)}_{i,j}$ is simply the total number of ways of pairing $m$ objects, so
\begin{equation}\label{eq:ansum}
\sum_{j=0}^{i}a^{(m)}_{i,j}\ =\ \begin{cases} (2n-1)!! & \text{if }m=2n;\\[3mm]
(2n+1)!! & \text{if }m=2n+1. \end{cases}
\end{equation}

The expressions for $a^{(m)}_{i,j}$ may be rewritten in a more convenient form, which separates, as far as possible, factors depending only on $i$ from those depending only on~$j$.  Starting with definition~\eqref{eq:Adefn} and assuming $i$ and $j$ have the same parity, 
\begin{align*}
a^{(m)}_{i,j}&=(2n)!\,\frac{i!(2n-i)!}{(2n)!}\;\frac1{j!}\;\frac{(i-j-1)!!}{(i-j)!}\;\frac{(2n-i-j-1)!!}{(2n-i-j)!}\\
&=(2n)!{\binom{2n}{i}}^{-1}\frac1{j!(n-j)!}\;\frac{(n-j)!}{2^{(i-j)/2}[(i-j)/2]!\; 2^{n-(i+j)/2}[n-(i+j)/2]!}\\
&=\frac{(2n)!}{2^nn!}{\binom{2n}{i}}^{-1}2^j\binom nj\binom{n-j}{(i-j)/2}.
\end{align*}
Thus we reach the following equivalent form of~\eqref{eq:Adefn}, valid for $m=2n$ even:
 \begin{align}
	a^{(m)}_{i,j}\, &=\,
	\begin{cases}\ds
		(2n-1)!!\,{2n \choose i}^{-1}  2^j{n \choose j}{n - j \choose (i-j)/2} & \text{if } i \equiv j \bmod 2\;,\\[3mm]
		\ 0 & \text{otherwise}\;.
	\end{cases}\label{eq:an1}
\intertext{Starting with \eqref{eq:Adefn'} and performing very similar manipulations we obtain the following, valid for $m=2n+1$ odd:}
a^{(m)}_{i,j}\, &=\,
(2n+1)!!\,{2n+1 \choose i}^{-1}  2^j{n \choose j}{n - j \choose \lfloor(i-j)/2\rfloor} \,. \label{eq:an2}
\end{align}
Clearly $a^{(m)}_{i,j}=0$ for $i<j$, so $A_m$ is lower triangular. (We use the usual convention that ${b\choose a}=0$ if $a<0$ or $b<a$.)

A \emph{pinning} of a function $F:\{0,1\}^d\to\reals^+$ is a function $G:\{0,1\}^m\to\reals^+$ of smaller or equal arity that is obtained from $F$ by setting certain variables to 0 or 1.  If $F$ is symmetric, then the it does not matter which variables are set to 0 (or 1), only their number.  Thus, a pinning of~$F$ may be written 
$$G(x_1,\ldots,x_m)=F(x_1,\ldots,x_m,\underbrace{0,\ldots,0}_{d-b},\underbrace{1,\ldots,1}_a),$$
where $0\leq a\leq b\leq d$ and $m=b-a$.  Note that, if the signature vector of $F$ is $[f_0,\ldots,f_d]$, then that of $G$ is $[f_a,\ldots,f_b]$.
Now, to prove that $F_{\beta, \mu, d}$ is windable, we use the following result from~\cite{HLZ16}.
\begin{theorem}[\cite{HLZ16}, Theorem 7]\label{thm01}
A symmetric function $F : \{0,1\}^d \mapsto \reals^+$ is \emph{windable} if and only if, for every $1\leq m\leq d$ and every pinning $G$ of $F$ with arity $m$, the self-complementary function $H(x) \eqdef G(x)G(\bar{x})$ with half-vector $\bs{h}=[h_0,h_1,\ldots,h_n]$ (where $n=\lfloor m/2\rfloor$) satisfies the following condition: The linear equations $A_m \bs{x} = \bs{h}$ have a nonnegative solution $\bs{x}\in (\reals^+)^{n+1}$.
\end{theorem}
To apply this, we first determine $H(x)=G(x)G(\bar{x})$ for all pinnings $G$ of $F_{\beta,\mu, d}$. To help guide the reader through the remainder of the section, the naming conventions in Theorem~\ref{thm01} will be maintained throughout.  So, $d$ will be the arity of some function~$F$, integer $m$ with $1\leq m\leq d$ the arity of some self-complementary function~$H$ derived from a pinning of~$F$, and $n=\lfloor m/2\rfloor$.  Thus the indices of the half-vector~$\bs{h}$ representing~$H$ run from 0 to~$n$.

\begin{lemma}\label{lem01}
	Let $d \in\pnats$ and $\beta,\mu\in\reals$, and let $G$ be a pinning of $F_{\beta,\mu, d}$. Define $H$ as in Theorem~\ref{thm01}.  Then, there exists $m\leq d$ and a constant $K \geq 0$ such that $H = K \cdot F_{2\beta, m}$.
\end{lemma}
\begin{proof}
	Let $[f_0, \ldots, f_{d}]$ be the signature vector of $F_{\beta,\mu, d}$, where $f_i = e^{\beta i(d-i)+\mu i}$. Let $G$ be a pinning of $F$, with $a$ variables pinned to 1, and $d-b$ to 0. Its signature vector is of the form $[f_i:a \leq i \leq b]$. Let $[z_i:0 \leq i \leq m]$ be the signature vector of~$H$, where $m=b-a\leq d$. Then we have:
	\begin{align*}
		z_i\ &=\ e^{\beta (a+i) (d - a - i) +\mu(a+i)} \cdot e^{\beta(b - i)(d- b + i)+\mu(b-i)}\ =\ e^{Z_i}
	\intertext{where the exponent $Z_i$ is:}
		Z_i\ &=\  \beta\big((a+i) (d - a - i)  + (b - i)(d- b + i)\big) + \mu(a+b)\\
		&=\  \beta(a(d-a)  + b(d-b))  + \mu(a+b) + 2\beta i(b-a-i).
	\intertext{	If $m=b-a$ and $K=e^{\beta(a(d-a)  + b(d-b))  + \mu(a+b)}$, this implies that we have:}
	z_i\ &=\ K e^{2\beta i(m-i)} = K F_{2\beta,m}\,.
	\end{align*}
	Thus $H$ is a positive multiple of the function $F_{2\beta, m}$.
\end{proof}
 
From Theorem~\ref{thm01} and Lemma~\ref{lem01}, it follows that
\begin{corollary}
	\label{cor:wind-2decomp}
	$F_{\beta,\mu, d}$ is windable if and only if for all $1\leq m\leq d$ the following holds:  the system $A_m\bs{x}=\bs{h}$ has a solution with $\bs{x}\geq \mathbf{0}$, where $\bs{h}$ is the half vector representing the function $F_{2\beta,m}$.
\end{corollary}
We can put the equations $A_m\bs{x}=\bs{h}$ in a more convenient form. Let
\begin{equation*} \label{eq:bn0}
\ds	b^{(m)}_{i,j}\ =\ \begin{cases} \ds
		{n - j \choose \lfloor(i-j)/2\rfloor}&
		\ds \text{if\ }d=2n+1 \text{\ or\ }i\equiv j \bmod 2\\[3mm]
	\ 0 & \text{otherwise.}
	\end{cases}
\end{equation*}
It follows from (\ref{eq:an1},\ref{eq:an2}) that the equations $A_m\bs{x}=\bs{h}$ can be rewritten in the form
\begin{equation}\label{eq:bnsum}
\sum_{j=0}^n b^{(m)}_{i,j}y_j = \ \begin{cases}\ds \frac{1}{(2n-1)!!} {m \choose i}h_i& \text{if }m=2n,\\[3mm]
\ds \frac{1}{(2n+1)!!}{m \choose i}h_i & \text{if }m=2n+1, \end{cases}
\end{equation}
where $y_j= 2^j{n \choose j}x_j$ for all $j\in[0,n]$.
Let $B_m$ be the $n\times n$ matrix $(b^{(m)}_{i,j}:i,j\in[0,n])$, and let $\bs{b}_j$ be its $j$th column, so $B_m=[\bs{b}_0\ \bs{b}_1\ \cdots\ \bs{b}_n]$. Then we can express the equations $A_m\bs{x}=\bs{h}$ in a form which suppresses irrelevant scalars.

The (convex) \emph{cone} $\C(S)\subseteq\reals^{n+1}$ generated by the set $S\subseteq\reals^{n+1}$ is the closure of $S$ under addition and multiplication by non-negative scalars. That is, for all $\bs{x}\in\C(S)$, $\kappa\bs{x}\in\C(S)$ for any $\kappa\geq 0$ and, for all $\bs{x},\bs{y}\in\C(S)$, $\bs{x}+\bs{y}\in\C(S)$. Clearly we always have $\mathbf{0}\in\C(S)$. Also, establishing $\bs{x}\in\C(S)$ is clearly equivalent to establishing $\kappa\bs{x}\in\C(S)$ for any $\kappa> 0$. Then we may restate Theorem~\ref{thm01} as
\begin{theorem}\label{thm02}
Let $F : \{0,1\}^d \mapsto \reals^+$ be a symmetric function and, let $G$ be a pinning of $F$ with arity $m$. Let $H(x) \eqdef G(x)G(\bar{x})$ with half-vector $\bs{h}=[h_0,h_1,\ldots,h_{n}]$, where $n\eqdef\lfloor m/2\rfloor$. Let $C_{n}\subset\reals^{n+1}$ be the cone generated by the columns of $B_{m}$, and let $\bs{z}(\bs{h})$ be the vector such that $\bs{z}(\bs{h})_i={m\choose i}h_i$ $(i\in[0,n])$. Then $F$ is windable if and only if $\bs{z}(\bs{h})\in C_{n}$ for every such $G$.
\end{theorem}
Lemma~\ref{lem01} implies that we need only apply this criterion to functions $H$ of the form $F_{\beta,d}$.

From~\eqref{eq:bnsum}, we have $\bs{z}(\mathbf{1})\in C_n$, where $\mathbf{1}$ is the all-1's vector.  (To see this, substitute $\bs{x}=\mathbf{1}$ 
in~\eqref{eq:ansum} to give $A\mathbf{1}=\bs{h}$, where $\bs{h}=(2n\pm1)!!\,\mathbf{1}$.  Now transform variables from $\bs{x}$ to $\bs{y}$ and compare with~\eqref{eq:bnsum}.)
Observe that $b^{(m)}_{i,j}=b^{(m-2k)}_{i-k,j-k}$, for all $k\in[0,n]$.  (Informally, deleting the first row and column of the matrix $B_m$ yields the matrix $B_{m-2}$.)  Then it follows that there exists a vector of the form 
$$
\hat{\bs{y}}=(\underbrace{0,0,\ldots,0}_{\text{$k$ copies}},y_0,y_1,\ldots,y_{n-k}) 
$$
such that
\begin{equation}\label{eq:bksum}
\sum_{j=0}^n b^{(m)}_{i,j}\yhat_j=\sum_{j=k}^n b^{(m)}_{i,j}y_{j-k}\ =\ \sum_{j=0}^{n-k} b^{(m-2k)}_{i-k,j}y_j \ =\ \begin{cases}\ds \frac{1}{(2n-2k-1)!!} {m-2k \choose i-k}& \text{if\ }m=2n,\\[4mm]
\ds \frac{1}{(2n-2k+1)!!}  {m-2k \choose i-k} & \text{if }m=2n+1. \end{cases}\,
\end{equation}
Specifically, we can choose $(y_0,y_1,\ldots,y_{n-k})$ to be the vector that witnesses, via~\eqref{eq:bnsum}, the membership of $\bs{z}(\mathbf{1}_{n-k+1})$ in the cone $C_{n-k}$.
Thus, for all $k\in[0,n]$, the vector
\[ \bs{v}_k=\left[{m-2k \choose i-k}:i\in [0,n]\right]\]
is in the cone~$C_n$.  Let $D_n$ be the cone generated by the set $\{\bs{v}_0,\bs{v}_1,\ldots,\bs{v}_n\}$. Then $D_n\subseteq C_n$.
Thus, to show $\bs{z}(\bs{h})\in C_n$, it suffices to show that $\bs{z}(\bs{h})\in D_n$.
\begin{lemma}\label{lem02}
$\bs{z}(\bs{h})\in D_n$, where $\bs{h}$ is the half-vector representing $F_{\beta,m}$
\end{lemma}
\begin{proof}
Expanding the exponential function in the definition of~$h_i$,
\[  \bs{z}(\bs{h})_i\ =\  {m\choose i}h_i\ =\ \sum_{j=0}^\infty \frac{\beta^j}{j!} \big(i(m-i)\big)^j{m \choose i}\qquad (i\in[0,n])\,.\]
Thus, the claim $\bs{z}(\bs{h})\in D_n$ will follow from showing that the vector with entries $\ds \big(i(m-i)\big)^j{m \choose i}$\quad $(i\in[0,n])$ is in $D_n$ for all $j\in\nats$.

Now, since $\bs{v}_0=\bs{z}(\mathbf{1})\in D_n$, the inclusion $\bs{z}(\bs{h})\in D_n$ will follow, by induction on $j$, from the following

\textbf{Claim}: If $\bs{u}\in D_n$, then the the vector $\bs{u}^*$ with entries $u^*_i=i(m-i)u_i$ is in $D_n$.

\emph{Proof}.
It suffices to prove this for the generators of $D_n$, namely $\bs{v}_k=[v_{0,k},v_{1,k},\ldots,v_{n,k}]$. Then,
if $\bs{v}_{-1}\eqdef \mathbf{0}$,
\begin{align*}
      i(m-i)v_{i,k} &= \big[(i-k)(m-k-i)+k(m-k)\big]{m-2k \choose i-k} \\
       &= \frac{(m-2k)!}{(i-k-1)!(m-k-i-1)!}+ k(m-k){m-2k \choose i-k} \\
       &= (m-2k)(m-2k-1)\frac{(m-2k-2)!}{(i-k-1)!(m-k-i-1)!}+ k(m-k){m-2k \choose i-k} \\
       &= (m-2k)(m-2k-1)v_{i,k+1}+ k(m-k)v_{i,k}
    \end{align*}
So $\bs{v}^*_k=(m-2k)(m-2k-1)\bs{v}_{k+1}+k(m-k)\bs{v}_k\in \C(\bs{v}_{k+1},\bs{v}_k)\subseteq D_n$. 
Note that the scalars appearing in the above linear combination of $\bs{v}_{k+1}$ and $\bs{v}_{k}$ are non-negative: if $m=2n$ is even and $k=n$ then factor $m-2k-1$ is negative, but in this case $m-2k=0$.
\end{proof}

Theorem~\ref{thm:main} follows by combining Lemma~\ref{lem01}, Theorem~\ref{thm02}, Lemma~\ref{lem02} and the fact that $D_n\subseteq C_n$.   

\section{Positive results for approximate computation} \label{sec:pos}

It is now a short step to the two results claimed in the Introduction.  Note that in this section $n=|V(\Gamma)|$ and $m=|E(\Gamma)|$ will always denote the number of vertices and edges of an instance graph~$\Gamma$.

\subsection{Proofs of the main results}

\begin{proof}[Proof of Theorem~\ref{thm:FPAUS}]
This result follows directly from \cite[Lemma 29]{HLZ16} using Theorem~\ref{thm:main} and Lemma~\ref{lem:ratio}.  Note that the quantity $\mu_\Lambda(\Omega_0)$ appearing in that lemma is just $\hol_0(\Gamma)/(\hol_0(\Gamma)+\hol_2(\Gamma))$ in our notation, and hence $\mu_\Lambda(\Omega_0)^{-1}=O(m^2)$.
\end{proof}

\begin{proof}[Proof of Theorem~\ref{thm:FPRAS}]
The reduction from estimating the partition function to sampling configurations is standard.  Choose a suitably spaced sequence of values $0=\beta_0<\beta_1<\cdots<\beta_r=\beta$.  The ratios $Z_{\beta_i,\nu}(\Gamma)/Z_{\beta_{i-1},\nu}(\Gamma)$, for $1\leq i\leq r$, may be estimated from polynomially many samples from the distribution $\mathcal{D}_{\Gamma,\beta_{i-1},\nu}$ using the Markov chain described in Section~\ref{sec:overview}, which mixes in polynomial time by Theorem~\ref{thm:FPAUS}.  The partition function $Z_{0,\nu}(\Gamma)$ factorises, and is hence easy to compute.  Thus, $Z_{\beta,\nu}(\Gamma)$ itself may be estimated from the telescoping product
$$
Z_{\beta,\nu}(\Gamma)=\frac{Z_{\beta_r,\nu}(\Gamma)}{Z_{\beta_{r-1},\nu}(\Gamma)}\times 
\frac{Z_{\beta_{r-1},\nu}(\Gamma)}{Z_{\beta_{r-2},\nu}(\Gamma)}\times\cdots\times
\frac{Z_{\beta_1,\nu}(\Gamma)}{Z_{\beta_0,\nu}(\Gamma)}\times Z_{0,\nu}(\Gamma).
$$
Details can be found in \cite{SVV}, together with refinements.  

Alternatively, one could show self-reducibility of the partition function and apply \cite[Thm~3]{HLZ16}.  Details of this approach can be found in \cite[\S3.2]{mono}.
\end{proof}

\subsection{Glauber dynamics}\label{sec:Glauber}

Theorem~\ref{thm:FPAUS} was established by considering a Markov chain based on updates to half-edges of~$\Gamma$ (see Section~\ref{sec:overview}).  A more natural approach would be to work with whole edges.  So, the configurations would simply be the natural ones, namely $\sigma:E(\Gamma)\to\{0,1\}$, and a transition would update a single (whole) edge.  We might call this the (single-site) Glauber dynamics.  As usual, it is easy to arrange transition probabilities so that the stationary distribution of the Markov chain is the Gibbs distribution.  

A convenient way to define Glauber dynamics for our purposes is the following.  Start with the half-edge Markov chain on $\Omega=\Omega_0\cup\Omega_2$ as defined in~\eqref{eq:Metropolis}, and censor all transitions from a state in~$\Omega_0$ to a state in $\Omega_2$.  That is to say, instead of making a transition from $\tau\in\Omega_0$ to $\tau'\in\Omega_2$, the Markov chain remains at~$\tau$.  Since the number of available transitions out of any state is now $m$ instead of $\binom{2m}2$, we take the opportunity to increase the proposal probability from $1/4m^2$ to $1/2m$.  This simple modification speeds up the Markov chain and reduces the mixing time by a factor~$2m$.  Assuming the initial state is in~$\Omega_0$, the censored Markov chain on state space~$\Omega_0$ is ergodic with stationary distribution $\calD_{\Gamma,\beta,\nu}$.  Note that every transition of this derived Markov chain has the following property:  the two half-edges whose spins are flipped belong to the same whole edge.  Since $\Omega_0$ is isomorphic to $2^{E(\Gamma)}$ it would have been more natural to define the Glauber dynamics directly on whole-edge configurations, but the above view simplifies the analysis that follows. 

To show that Glauber dynamics is rapidly mixing, we will compare it with the already analysed half-edge Markov chain.  To follow the argument in detail, some acquaintance with the canonical paths method is required, and in particular the use made of it by Huang et al.~\cite{HLZ16}.  However, it should be possible to obtain an intuitive grasp of the argument, without that prerequisite.  Denote the transition probabilities of the half-edge Markov chain by $P:\Omega^2\to[0,1]$ and the stationary distribution by $\pi:\Omega\to[0,1]$.  Imagine the Markov chain as a transition graph with states as vertices and transitions as edges.  Each edge (transition) $(\tau,\tau')$ is assigned a capacity $Q(\tau,\tau')=\pi(\tau)P(\tau,\tau')$.  Since the Markov chain is time-reversible, $Q(\tau,\tau')=Q(\tau',\tau)$ and the capacity is well defined.  For each pair of states $\sigma,\sigma'$, we specify a collection of (canonical) paths from $\sigma$ to $\sigma'$ through which we route a flow of $\pi(\sigma)\pi(\sigma')$.  Suppose we can do this in such a way that the total flow through any transition $\tau\to\tau'$ is small, specifically, less than or equal to $\rho\,Q(\tau,\tau')$ for some uniform bound~$\rho$.  Then the relaxation time of the Markov chain (a quantity strongly related to the mixing time) is known to be $O(\rho)$.  The quantity $\rho$ is the \emph{congestion} of the chosen canonical paths.  For more detail see, e.g., \cite[\S5.2]{mono}.

The canonical paths exploited first by McQuillan and later by Huang et al., are derived from collection of (graph theoretic) paths and cycles in the instance graph~$\Gamma$.  Given such a collection, the corresponding canonical path in the transition graph is obtained by ``unwinding'' these paths and cycles in sequence.  Unwinding consists in tracing along the path or around the cycle, at each vertex flipping the two half-edges incident at that vertex.  (See \cite[\S B]{HLZ16} for details.)  Suppose that $e_0,e_1,e_2,e_3,\ldots,e_{2\ell-2},e_{2\ell-1}$ is a path or cycle in $\Gamma$ of length~$\ell$, composed of half-edges.  For each $0\leq i\leq\ell-1$, the half edges $e_{2i},e_{2i+1}$ form a single edge, and for each $1\leq i\leq\ell-1$, the half edges $e_{2i-1},e_{2i}$ are incident at a common vertex.  The half edges $e_0$ and $e_{2\ell-1}$ may or may not be incident at a common vertex.  
In Huang et al's analysis, the half-edges are flipped in the sequence $(e_{2\ell-1},e_0),(e_1,e_2),\ldots,(e_{2\ell-3},e_{2\ell-2})$.  If we pair the half edges in the other natural way, i.e., $(e_0,e_1),(e_2,e_3),\ldots,(e_{2\ell-2},e_{2\ell-1})$, then we obtain a canonical path that flips whole edges at each step.  In this way, a canonical path on half-edges induces a canonical path on (whole) edges.  As we shall see, these canonical paths witness rapid mixing of the Glauber dynamics.  It is important in this context to note that the signatures we are working with are all permissive, in the sense that each function $f_v$ is supported on the whole of $\{0,1\}^{E(v)}$.  In other situations, the derived paths might be invalid.

In more detail, suppose    
$$
\sigma=\sigma_0\to\sigma_1\to\cdots\to\sigma_i\to\sigma_{i+1}\to\cdots\to\sigma_\ell=\sigma'
$$
is the canonical path in $\Omega_0\cup\Omega_2$ from $\sigma\in\Omega_0$ to $\sigma'\in\Omega_0$ that results from the sequence $(e_{2\ell-1},e_0),(e_1,e_2),\ldots,(e_{2\ell-3},e_{2\ell-2})$ of half-edge flips.  Then let 
$$
\sigma=\sigmatilde_0\to\sigmatilde_1\to\cdots\to\sigmatilde_i\to\sigmatilde_{i+1}\to\cdots\to\sigmatilde_\ell=\sigma'
$$
be the derived path in $\Omega_0$ resulting from the sequence $(e_0,e_1),(e_2,e_3),\ldots,(e_{2\ell-2},e_{2\ell-1})$.  Thus, for $0\leq i\leq\ell-1$,
$$
\sigmatilde_{i+1}(e)=
\begin{cases}
1-\sigmatilde_{i}(e),&\text{if $e\in\{e_{2i},e_{2i+1}\}$};\\
\sigmatilde_i(e),&\text{otherwise}.
\end{cases}
$$
The derived path stays close to the original path at all times;  specifically, for $1\leq i\leq\ell-1$,
\begin{equation}\label{eq:sigmatosigmatilde}
\sigmatilde_i(e)=\begin{cases}
1-\sigma_i(e),&\text{if $e\in\{e_{2i-1},e_{2\ell-1}\}$};\\
\sigma_i(e),&\text{otherwise}.
\end{cases}
\end{equation}
 
Consider a transition $\sigmatilde_i\to\sigmatilde_{i+1}$ in the derived path.  By comparing $\sigmatilde_i$ and $\sigmatilde_{i+1}$ we may deduce the unordered pair of edges $\{e_{2i},e_{2i+1}\}$, and with one extra bit of information we may identify the edges $e_{2i}$ and $e_{2i+1}$ individually.  If we now specify $e_{2i-1}$ and $e_{2\ell-1}$, we may recover the corresponding transition $\sigma_i\to\sigma_{i+1}$ in the original path, by putting \eqref{eq:sigmatosigmatilde} into reverse.  (The first and last transitions are special, but the given data is still sufficient to recover the transition $\sigma_i\to\sigma_{i+1}$.)  Since $e_{2i-1}$ shares a vertex with edge $e_{2i}$, we have at most $n$~choices for $e_{2i-1}$, where $n=|V(\Gamma)|$.  The upshot of this is that, given the transition $\sigmatilde_i\to\sigmatilde_{i+1}$, there are at most $2nm$ possibilities for the corresponding transition $\sigma_i\to\sigma_{i+1}$.  Suppose $\tautilde\to\tautilde'$ is any transition with $\tautilde,\tautilde'\in\Omega_0$.  Define the \emph{halo} around $\tautilde\to\tautilde'$ to be the set of $2nm$ transitions $\tau\to\tau'$ that correspond to $\tautilde\to\tautilde'$ via the above procedure.  Note that the flow in any derived path through $\tautilde\to\tautilde'$ must originally have passed through some transition in this (small) halo.

Now observe that the capacity of any transition $\tau\to\tau'$ in the halo of transition $\tautilde\to\tautilde'$ is not too much bigger that the capacity of $\tautilde\to\tautilde'$.  Since $\tau$ differs from $\tautilde$ at at most two half-edges, the weights $\what(\tau)$ and $\what(\tautilde)$ are within a constant factor of each other; the same applies to $\what(\tau')$ and $\what(\tautilde')$.  It follows that $Q(\tau,\tau')\leq CQ(\tautilde,\tautilde')$, where $C=C(\beta,\nu,\Delta)$.  Suppose $\rho$ is the (known) congestion of the half-edge Markov chain.  Then the flow in transition $\tau\to\tau'$ is bounded by $\rho\,Q(\tau,\tau')\leq C\rho\,Q(\tautilde,\tautilde')$.  So if we reroute the flow between pairs of states in $\Omega_0$ through the derived paths, the new flow in transition $\tautilde\to\tautilde'$ is bounded by  $2nm\times C\rho\,Q(\tautilde,\tautilde')$.  So rerouting flow through paths lying entirely within $\Omega_0$ increases congestion by at most a factor $2Cnm$.  

Since the derived paths lie entirely within~$\Omega_0$, they can be used to bound the congestion of the censored Markov chain (Glauber dynamics). Denote the transition probabilities of this full-edge Markov chain by $\Ptilde:\Omega_0^2\to[0,1]$ and the stationary distribution by $\pitilde:\Omega_0\to[0,1]$.  Also define the capacities of transitions in the censored Markov chain by $\Qtilde(\tau,\tau')=\pitilde(\tau)\Ptilde(\tau,\tau')$.  Let $R=\pi(\Omega_0)^{-1}$, and note that $\pitilde(\tau)=R\pi(\tau)$ and 
$$\Qtilde(\tau,\tau')=\pitilde(\tau)\Ptilde(\tau,\tau')=R\pi(\tau)\times 2mP(\tau,\tau')=2Rm Q(\tau,\tau'),$$ 
for all $\tau,\tau'\in\Omega_0$.  Thus the flows are increased by a factor $R^2$ but the capacities are increased by a factor $2Rm$.  Thus the congestion $\rhotilde$ of the censored Markov chain is increased by a factor $R/2m$.  Putting it all together, $\rhotilde=2Cnm\rho R^2/2mR=CnR\rho$.  Since $R=O(m^2)$, we see that the relaxation time (and hence the mixing time) of Glauber dynamics is at most $O(nm^2)$ longer than that of the half-edge Markov chain.  The comparison argument could be tightened by delving further into the details of the canonical paths analysis in~\cite{HLZ16}.  

\subsection{Ferromagnetic interactions}
Direct calculation from Theorem~\ref{thm01} shows that $[1,e^{2\beta},e^{2\beta},1]$ is windable when $e^{4\beta}\geq\frac12$. Thus, Glauber dynamics is rapidly mixing on line graphs of cubic graphs (such as fragments of the kagome lattice) when $\beta>-0.173286$.  The kagome lattice exhibits a phase transition when $e^{4\beta}=1/(3+2\sqrt3)$~\cite{Syozi}, so we cannot expect Glauber dynamics to be rapidly mixing when $\beta<-0.466567$.

Rapid mixing can also be deduced from Theorem~\ref{thm01} for graphs of maximum degree $\Delta$, for small values of $\Delta>3$ and all $\beta$ in some range $\beta_\Delta\leq \beta<0$.  It seems likely that a general argument could be found covering all $\Delta$, though we have not done this.  However, $\beta_\Delta$ will necessarily converge to~0 as $\Delta\to\infty$.  For $\Delta$ even, consider the star $\Gamma=K_{1,\Delta}$, the line graph of which is the complete graph~$K_\Delta$.  Let $\Omega$ denote the set of all configurations $E(\Gamma)\to\{0,1\}$ and $\Omega_{1/2}$ denote the subset of balanced configurations with equal numbers of 0s and 1s.  Then 
$$
\sum_{\sigma\in\Omega}\what(\sigma)\geq2\qquad\text{and}\qquad \sum_{\sigma\in\Omega_{1/2}}\what(\sigma)\leq2^\Delta\exp(\beta\Delta^2/4),
$$ 
where the first inequality comes from considering just the all 0s and all 1s configurations.  When $\beta<-4/\Delta$, the second sum is exponentially smaller than the first.  Since $\Omega_{1/2}$ forms a cut in the state space, a standard conductance bound \cite[Thm~2]{flow} then shows that mixing time of Glauber dynamics is exponential in~$\Delta$.

\subsection{Local fields}
Finally, varying fields, with an individual strength~$\nu_k$ for each vertex~$k$ of $L(\Gamma)$ can be handled.  One way is to repeat the calculations in Sections \ref{sec:holants} and \ref{sec:windability} with $\nu_k$ and~$\mu_k$ in place of~$\mu$ and $\nu$.  The notational complexity increases, but there are no essential changes.  However, the same end can be achieved with less effort as follows.  Subdivide each edge of $\Gamma$ by a single vertex with signature vector $[1,0,e^{\nu_k}]$. This signature is windable by \cite[Lemma 21]{HLZ16}, and correctly implements the field acting on a single spin.  The Markov chain on the extended configurations is rapidly mixing, as before.  

\section{Negative results for exact computation}\label{sec:neg}

The problem $\Hol([x_0,x_1,x_2,x_3])$ is the following.  Given a cubic graph~$\Gamma$, evaluate the holant~\eqref{eq:holant}, where $f_k$ is the function with signature vector $[x_0,x_1,x_2,x_3]$, for all $k\in V(\Gamma)$.  The problem $\PlHol([x_0,x_1,x_2,x_3])$ is the same with the instance~$\Gamma$ restricted to be planar.
According to the correspondence derived in Section~\ref{sec:holants}, the partition function of the antiferromagnetic Ising model on a line graph of a cubic graph can be expressed as an instance of $\Hol([1,a,a,1])$ with $a>1$.\footnote{To avoid issues with the representation of real numbers, and so obtain a clean computational problem, we introduce $a=e^{2\beta}$ and assume $a\in\mathbb{Q}$.}  Lemma~\ref{lem:zerofield} thus implies that computing \emph{exactly} the partition function of the antiferromagnetic Ising model on a graph~$\Gamma$ is $\numP$-hard, even in the absence of an external field, and even when $\Gamma$ is the line graph of a cubic graph.  The following results exploit a tiny fraction of a vast theory of the computational complexity of holant problems.  To see the wider scope, refer to Cai and Chen's monograph~\cite{CaiChenBook}.

\begin{lemma}\label{lem:zerofield}
$\Hol([1,a,a,1])$ is $\numP$-hard, for all $a\in \mathbb{Q}\cap(1,\infty)$.
\end{lemma}

The route to proving this lemma is via an intermediate computational problem.  Suppose $\Gamma'$ is a bipartite graph with vertices of degree two on one side of the bipartition and of degree three on the other.   The problem $\Hol([y_0,y_1,y_2]\mid[1,0,0,1])$ is the following.  Given a graph~$\Gamma'$ of the specified form, evaluate the holant~\eqref{eq:holant}, where vertices of degree two have signature vector $[y_0,y_1,y_2]$ and those of degree three have signature vector $[1,0,0,1]$ (i.e., ternary equality).  Following Cai, Huang and Lu,~\cite[\S3]{CHL12} we derive (possibly complex) numbers $y_0,y_1,y_2$ such that $\Hol([x_0,x_1,x_2,x_3])$ is polynomial-time interreducible with $\Hol([y_0,y_1,y_2]\mid[1,0,0,1])$.  We then complete the proof of Lemma~\ref{lem:zerofield} by invoking the following dichotomy result. 

\begin{lemma}\cite[Theorem 2]{CHL12}. \label{lem:equiv} 
The problem $\Hol([y_0,y_1,y_2]\mid[1,0,0,1])$ is $\numP$-hard for all $y_0,y_1,y_2\in\mathbb{C}$, except in the following cases, for which the problem is in $\mathrm{P}$: (1)~$y_1^2=y_0y_2$; (2)~$y_0^{12}=y_1^{12}$ and $y_0y_2=-y_1^2 \>(y_1\not=0)$; (3)~$y_1=0$; and (4)~$y_0=y_2=0$.
\end{lemma}

\begin{proof}[Proof of Lemma~\ref{lem:zerofield}]
In the notation of~\cite[\S3]{CHL12}, we have $x_0=x_3=1$ and $x_1=x_2=a$.  The first step is a change of variables.  Write
\begin{align*}
x_0&=\alpha_1^3+\beta_1^3\\
x_1&=\alpha_1^2\alpha_2+\beta_1^2\beta_2\\
x_2&=\alpha_1\alpha_2^2+\beta_1\beta_2^2\\
x_3&=\alpha_2^3+\beta_2^3,
\end{align*}
where $\alpha_1,\alpha_2,\beta_1,\beta_2$ are to be determined.  Setting
$\alpha_1=\beta_2=e^{i\theta}$ and $\alpha_2=\beta_1=e^{-i\theta}$ we have
\begin{align*}
x_0&=2\cos3\theta\\
x_1&=2\cos\theta\\
x_2&=2\cos\theta\\
x_3&=2\cos3\theta.
\end{align*}
Then, by setting $0<\theta<\frac\pi6$ appropriately, the signature vector $[x_0,x_1,x_2,x_3]$ can be made to match $[1,a,a,1]$ up to a scaling factor.  
Following Cai et al., introduce the matrix 
$$T=\begin{bmatrix}\alpha_1&\beta_1\\\alpha_2&\beta_2\end{bmatrix}=\begin{bmatrix}e^{i\theta}&e^{-i\theta}\\e^{-i\theta}&e^{i\theta}\end{bmatrix},$$
and (in the compact notation used in work on holants) define $[y_0,y_1,y_2]=[1,0,1]T^{\otimes2}=[2\cos2\theta,2,2\cos2\theta]$.  This last identity can be expanded into a less compact but more generally accessible matrix equation:
$$
\begin{bmatrix}y_0&y_1\\y_1&y_2\end{bmatrix}=T^\intercal I_2 T=\begin{bmatrix}e^{i\theta}&e^{-i\theta}\\e^{-i\theta}&e^{i\theta}\end{bmatrix}\begin{bmatrix}e^{i\theta}&e^{-i\theta}\\e^{-i\theta}&e^{i\theta}\end{bmatrix}=
\begin{bmatrix}2\cos2\theta&2\\2&2\cos2\theta\end{bmatrix}.
$$
(See the exercise preceding Theorem 1.6 of~\cite{CaiChenBook}.)  So we have $y_0=y_2=2\cos\theta$ and $y_1=2$.  Since none of the four conditions involving $y_0,y_1,y_2$ in the statement of Lemma~\ref{lem:equiv} hold, it follows that $\Hol([1,a,a,1])$ is $\numP$-hard.  
\end{proof}

We can even restrict the graph $\Gamma$ to be planar, at the expense of introducing a non-zero field.

\begin{lemma}\label{lem:planar}
$\PlHol([1,ab,ab^2,b^3])$ is $\numP$-hard, for all $a\in\mathbb{Q}\cap(1,\infty)$ and $b\in\mathbb{Q}^+\setminus\{0,1\}$.
\end{lemma}

The proof follows a similar line to the previous one, using the following dichotomy result for planar holants.    

\begin{lemma}  \label{lem:equivplanar} Kowalczyk and Cai~\cite[Thm~4.4]{KC16}.
The problem $\PlHol([y_0,y_1,y_2]\mid[1,0,0,1])$ is $\numP$-hard for all $y_0,y_1,y_2\in\mathbb{C}$, except in the four cases listed in the statement of Lemma~\ref{lem:equiv}, together with a fifth, namely (5)~$y_0^3=y_2^3$.  In these five cases, the problem is in $\mathrm{P}$. 
\end{lemma}

Note that in the statement of this result in~\cite[Thm~4.4]{KC16}, $\mathrm{Hol}(a,b)$ is shorthand for 
$\Hol([a,1,b]\mid[1,0,0,1])$.

\begin{proof} [Proof of Lemma~\ref{lem:planar}]
Define $\alpha_1=e^{-\lambda+i\theta}$, $\beta_1=e^{-\lambda-i\theta}$, $\alpha_2=e^{\lambda-i\theta}$ and $\beta_2=e^{\lambda+i\theta}$;  then define   
\begin{align*}
x_0&=\alpha_1^3+\beta_1^3=2e^{-3\lambda}\cos3\theta\\
x_1&=\alpha_1^2\alpha_2+\beta_1^2\beta_2=2e^{-\lambda}\cos\theta\\
x_2&=\alpha_1\alpha_2^2+\beta_1\beta_2^2=2e^{\lambda}\cos\theta\\
x_3&=\alpha_2^3+\beta_2^3=2e^{3\lambda}\cos3\theta.
\end{align*}
By setting $0<\theta<\frac\pi6$ and $\lambda\in\mathbb{R}$ appropriately, the signature vector $[x_0,x_1,x_2,x_3]$ can be made to match $[1,ab,ab^2,b^3]$ up to a scaling factor.  
Again define 
$$T=\begin{bmatrix}\alpha_1&\beta_1\\\alpha_2&\beta_2\end{bmatrix}=\begin{bmatrix}e^{-\lambda+i\theta}&e^{-\lambda-i\theta}\\e^{\lambda-i\theta}&e^{\lambda+i\theta}\end{bmatrix},$$
and 
$$[y_0,y_1,y_2]=[1,0,1]T^{\otimes2}=[e^{-2\lambda+2i\theta}+e^{2\lambda-2i\theta},\>e^{2\lambda}+e^{-2\lambda},\>e^{2\lambda+2i\theta}+e^{-2\lambda-2i\theta}].$$  
To establish $\numP$-hardness for planar graphs, we just need to verify that none of the five conditions in Lemma~\ref{lem:equivplanar} hold.  To this end, observe that 
$$y_0y_2=2\cosh4\lambda+2\cos4\theta\quad\text{and}\quad y_1^2=2\cosh4\lambda+2.$$
Noting that $\cosh4\lambda\geq1$ and $0<\theta<\frac\pi6$, we find that $0<y_0y_2<y_1^2$.  Also, $y_2$ is the complex conjugate of $y_0$ and so $|y_0|=|y_2|<|y_1|$.  These facts rule out the four polynomial-time cases (1)--(4) listed in the statement of Lemma~\ref{lem:equivplanar}.  The fifth condition is ruled out by $\arg y_0 =-\arg y_2\not=0$ and $-\frac\pi3<\arg y_0,\arg y_2<\frac\pi3$.  It follows that $\PlHol([1,ab,ab^2,b^3])$ is \#P-hard.
\end{proof}

Note that the final step of the proof crucially uses the fact that $\lambda\not=0$, which is equivalent to $b\not=1$.  Indeed, the partition function can be evaluated in polynomial time on planar graphs in the zero field case, using the algorithm of Fisher and Kasteleyn mentioned earlier.  Lemma~\ref{lem:planar} can be interpreted as asserting that exactly computing the partition function of the Ising model with a constant non-zero external field is hard even when restricted to line graphs of cubic planar graphs. An example of such a graph is the kagome lattice with toroidal boundary conditions, as it is the line graph of the hexagonal lattice with the same boundary conditions.

\section*{Acknowledgements}
MJ thanks Heng Guo, Alan Sokal and Eric Vigoda for helpful remarks.

\bibliographystyle{plain}
\bibliography{antiferroIsing}

\begin{thebibliography}{10}

\bibitem{Arbib88}
Claudio Arbib.
\newblock A polynomial characterization of some graph partitioning problems.
\newblock {\em Inform. Process. Lett.}, 26(5):223--230, 1988.

\bibitem{BCR20}
Ferenc Bencs, P\'eter Csikv\'ari, and Guus Regts.
\newblock Some applications of {W}agner's weighted subgraph counting
  polynomial.
\newblock arXiv:2012.00806, 2020.

\bibitem{BJ00}
Hans~L. Bodlaender and Klaus Jansen.
\newblock On the complexity of the maximum cut problem.
\newblock {\em Nordic J. Comput.}, 7(1):14--31, 2000.

\bibitem{CaiChenBook}
Jin-Yi Cai and Xi~Chen.
\newblock {\em Complexity dichotomies for counting problems. {V}ol.~1:
  {B}oolean domain}.
\newblock Cambridge University Press, Cambridge, 2017.

\bibitem{CHL12}
Jin-Yi Cai, Sangxia Huang, and Pinyan Lu.
\newblock From {H}olant to \#{CSP} and back: dichotomy for {H}olant{$^c$}
  problems.
\newblock {\em Algorithmica}, 64(3):511--533, 2012.

\bibitem{chen2020rapid}
Zongchen Chen, Kuikui Liu, and Eric Vigoda.
\newblock Rapid mixing of {G}lauber dynamics up to uniqueness via contraction.
\newblock {\em CoRR}, abs/2004.09083, 2020.

\bibitem{Fisher}
Michael~E. Fisher.
\newblock On the dimer solution of planar {I}sing models.
\newblock {\em Journal of Mathematical Physics}, 7(10):1776--1781, 1966.

\bibitem{FriedliVelenik}
S.~Friedli and Y.~Velenik.
\newblock {\em Statistical mechanics of lattice systems}.
\newblock Cambridge University Press, Cambridge, 2018.
\newblock A concrete mathematical introduction.

\bibitem{GJS76}
M.~R. Garey, D.~S. Johnson, and L.~Stockmeyer.
\newblock Some simplified {NP}-complete graph problems.
\newblock {\em Theoret. Comput. Sci.}, 1(3):237--267, 1976.

\bibitem{GMP}
Leslie~Ann Goldberg, Russell Martin, and Mike Paterson.
\newblock Random sampling of 3-colorings in {$\mathbb{Z}^{2}$}.
\newblock {\em Random Structures Algorithms}, 24(3):279--302, 2004.

\bibitem{Guruswami}
Venkatesan Guruswami.
\newblock Maximum cut on line and total graphs.
\newblock {\em Discrete Appl. Math.}, 92(2-3):217--221, 1999.

\bibitem{HLZ16}
Lingxiao Huang, Pinyan Lu, and Chihao Zhang.
\newblock Canonical paths for {MCMC:} from art to science.
\newblock In Robert Krauthgamer, editor, {\em Proceedings of the Twenty-Seventh
  Annual {ACM-SIAM} Symposium on Discrete Algorithms, {SODA} 2016, Arlington,
  VA, USA, January 10-12, 2016}, pages 514--527. {SIAM}, 2016.

\bibitem{mono}
Mark Jerrum.
\newblock {\em Counting, sampling and integrating: algorithms and complexity}.
\newblock Lectures in Mathematics ETH Z\"{u}rich. Birkh\"{a}user Verlag, Basel,
  2003.

\bibitem{JerSinIsing}
Mark Jerrum and Alistair Sinclair.
\newblock Polynomial-time approximation algorithms for the {I}sing model.
\newblock {\em SIAM J. Comput.}, 22(5):1087--1116, 1993.

\bibitem{JJ20}
E.~Jur\v{c}i\v{s}inov\'{a} and M.~Jur\v{c}i\v{s}in.
\newblock Ground states, residual entropies, and specific heat capacity
  properties of frustrated {I}sing system on pyrochlore lattice in effective
  field theory cluster approximations.
\newblock {\em Phys. A}, 554:124671, 13, 2020.

\bibitem{Kasteleyn}
P.~W. Kasteleyn.
\newblock Dimer statistics and phase transitions.
\newblock {\em J. Mathematical Phys.}, 4:287--293, 1963.

\bibitem{KC16}
Michael Kowalczyk and Jin-Yi Cai.
\newblock Holant problems for 3-regular graphs with complex edge functions.
\newblock {\em Theory Comput. Syst.}, 59(1):133--158, 2016.

\bibitem{LV99}
Michael Luby and Eric Vigoda.
\newblock Fast convergence of the {G}lauber dynamics for sampling independent
  sets.
\newblock {\em Random Structures Algorithms}, 15(3-4):229--241, 1999.
\newblock Statistical physics methods in discrete probability, combinatorics,
  and theoretical computer science (Princeton, NJ, 1997).

\bibitem{McQuil13}
Colin McQuillan.
\newblock Approximating {H}olant problems by winding.
\newblock {\em CoRR}, abs/1301.2880, 2013.

\bibitem{MotRag}
Rajeev Motwani and Prabhakar Raghavan.
\newblock {\em Randomized algorithms}.
\newblock Cambridge University Press, Cambridge, 1995.

\bibitem{flow}
Alistair Sinclair.
\newblock Improved bounds for mixing rates of {M}arkov chains and
  multicommodity flow.
\newblock {\em Combin. Probab. Comput.}, 1(4):351--370, 1992.

\bibitem{Syozi}
Itiro Sy\^{o}zi.
\newblock Statistics of {K}agom{\'e} lattice.
\newblock {\em Progr. Theoret. Phys.}, 6:306--308, 1951.

\bibitem{SVV}
Daniel \v{S}tefankovi\v{c}, Santosh Vempala, and Eric Vigoda.
\newblock Adaptive simulated annealing: a near-optimal connection between
  sampling and counting.
\newblock {\em J. ACM}, 56(3):Art. 18, 36, 2009.

\bibitem{WJ08}
Martin~J. Wainwright and Michael~I. Jordan.
\newblock Graphical models, exponential families, and variational inference.
\newblock {\em Foundations and Trends in Machine Learning}, 1(1-2):1--305,
  2008.

\end{thebibliography}

\end{document}